\documentclass[11pt]{article}
\usepackage{amsmath}
\usepackage{amsfonts}
\usepackage{amssymb}
\usepackage{amsthm}
\usepackage{graphicx}
\usepackage{fullpage}

\newcommand{\be}{\begin{equation}}
\newcommand{\ee}{\end{equation}}
\renewcommand{\[}{\begin{equation}}
\renewcommand{\]}{\end{equation}}
\newcommand{\ba}{\begin{eqnarray}}
\newcommand{\ea}{\end{eqnarray}}
\newtheorem*{thm}{Theorem}
\newtheorem{lemma}{Lemma}
\newcommand{\one}{\leavevmode\hbox{\small1\normalsize\kern-.33em1}}
\newcommand{\moy}[1]{\langle #1 \rangle}
\newcommand{\demi}{\frac{1}{2}}
\newcommand{\si}{\sigma}
\newcommand{\ket}[1]{\mbox{$ | #1 \rangle $}}
\newcommand{\braket}[2]{\mbox{$ \langle #1 | #2 \rangle $}}
\newcommand{\bra}[1]{\mbox{$ \langle #1 | $}}

\begin{document}

\title{Device-independent quantum key distribution\\
secure against collective attacks}
\author{Stefano Pironio$^{1}$\thanks{stefano.pironio@unige.ch}, Antonio Ac\'{\i}n$^{2,3}$, Nicolas Brunner$^4$,\\ Nicolas Gisin$^1$, Serge Massar$^5$, Valerio Scarani$^6$\\[0.5em]
$^1$ Group of Applied Physics, University of Geneva\\
$^2$ ICFO-Institut de Ciencies Fotoniques, 08860 Castelldefels, Spain\\
$^3$ ICREA-Instituci\'o Catalana de Recerca i Estudis
Avan\c cats\\
Pg. Lluis Companys 23, 08010 Barcelona, Spain\\
$^4$ H.H. Wills Physics Laboratory, University of Bristol\\
$^5$ Laboratoire d'Information Quantique, Universit\'{e} Libre de Bruxelles\\
C.P 225, Boulevard du Triomphe, B-1050 Bruxelles, Belgium \\
$^6$ Centre for Quantum Technologies and \\
Department of Physics,
National University of Singapore }
\date{}
\maketitle

\begin{abstract}
Device-independent quantum key distribution (DIQKD) represents a
relaxation of the security assumptions made in usual quantum key
distribution (QKD). As in usual QKD, the security of DIQKD follows
from the laws of quantum physics, but contrary to usual QKD, it
does not rely on any assumptions about the internal working of the
quantum devices used in the protocol. We present here in detail
the security proof for a DIQKD protocol introduced in [Phys. Rev.
Lett. 98, 230501 (2008)]. This proof exploits the full structure
of quantum theory (as opposed to other proofs that exploit the
no-signalling principle only), but only holds again collective
attacks, where the eavesdropper is assumed to act on the quantum
systems of the honest parties independently and identically at
each round of the protocol (although she can act coherently on her
systems at any time). The security of any DIQKD protocol
necessarily relies on the violation of a Bell inequality. We
discuss the issue of loopholes in Bell experiments in this
context.
\end{abstract}

\vfill
\pagebreak

\section{Introduction}
Device-independent quantum key distribution (DIQKD) protocols
aim at generating a secret key between two parties in a
provably secure way without making assumptions about the internal
working of the quantum devices used in the protocol. In DIQKD, the
quantum apparatuses are seen as black boxes that produce classical
outputs, possibly depending on the value of some classical inputs (see Fig.~\ref{figdiqkd}).
These apparatuses are thought to implement a quantum process, but
no hypothesis in terms of Hilbert space, operators, or states are
made on the actual quantum process that generates the outputs
given the inputs.

DIQKD can be thought of by contrasting it with usual quantum key
distribution (QKD). In its entanglement-based version
\cite{BBM92}, traditional QKD involves two parties, Alice and Bob,
who receive entangled particles emitted from a common source and
who measure each of them in some chosen bases. The measurement
outcomes are kept secret and form the raw key. As the source of
particles is situated between Alice's and Bob's secure locations,
it is not trusted by the parties, but assumed to be under the
control of an eavesdropper Eve. The eavesdropper could for
instance have replaced the original source by one who produces
states that give her useful information about Alice's and Bob's
measurement outcomes. However, by performing measurements in
well-chosen bases on a random subset of their particles and by
comparing their results, Alice and Bob can estimate the quantum
states that they receive from the eavesdropper and decide whether
a secret key can be extracted from them.

In a device-independent analysis of this scenario, Alice and Bob
would not only distrust the source of particles, but they would
also distrust their measuring apparatuses. The measurement
directions may for instance drift with time due to imperfections
in the apparatuses, or the entire apparatuses may be untrusted
because they have been fabricated by a malicious party. Alice and
Bob have therefore no guarantee that the actual measurement bases
corresponds to the expected ones. In fact they cannot even make
assumptions about the dimension of the Hilbert space in which they
are defined. In DIQKD, Alice and Bob have thus to bound Eve's
information by looking for the worst combination of states and
measurements (in Hilbert spaces of arbitrary dimension) that are
compatible with the observed classical input-output relations. In
contrast, in usual QKD Alice and Bob have a perfect knowledge of
the measurements that are carried out and of the Hilbert
space dimension of the quantum state they measure, and they exploit
this information to bound the eavesdropper's information when they
look for the worst possible states compatible with their observed
data.

\begin{figure}[h]\begin{center}
\vspace{2em}
\includegraphics[scale=0.65]{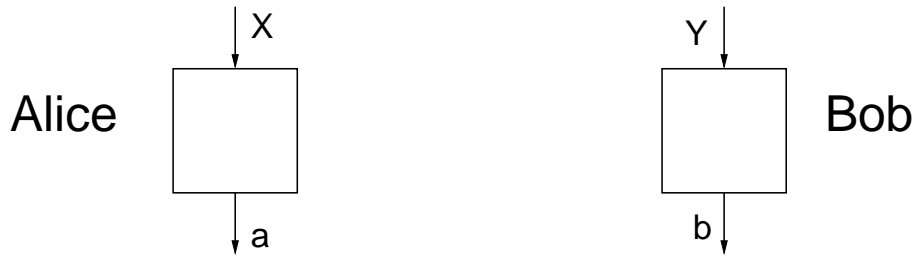}
\caption{Schematic representation of the DIQKD scenario. Alice and Bob see their quantum devices as black boxes producing classical outputs, $a$ and $b$, as a function of classical inputs $X$ and $Y$. From the observed statistics, and without making any assumption on the internal working of the devices, they should be able to conclude whether they can establish a secret key secure against a quantum eavesdropper.
}\label{figdiqkd}
\end{center}\end{figure}

\subsection{Why DIQKD?}
DIQKD represents a relaxation of the security assumptions made in
usual QKD. In this sense, it fits in the continuity of a series of
works that aim to design cryptographic protocols secure against
more and more powerful eavesdroppers.

From a fundamental point of view, DIQKD shows that the security of
a cryptographic scheme is possible based on a minimal set of
fundamental assumptions. It only requires that:
\begin{itemize}
\item Alice's and Bob's physical locations are secure, i.e., no unwanted information can leak out to the outside;
\item they have a trusted random number generator,
possibly quantum, producing a classical random output;
\item they have trusted classical devices (e.g., memories and computing devices) to store and process the classical data generated by their quantum
apparatuses;
\item they share an authenticated, but otherwise public, classical channel (this hypothesis can be ensured if Alice and Bob start off with a small shared secret key);
\item quantum physics is correct.
\end{itemize}
Other than these prerequisites, shared by all QKD protocols, no
others are necessary. In addition to these essential requirements,
usual QKD protocols assume that Alice and Bob have some knowledge
about their quantum devices.

From a practical point of view, DIQKD resolves some of the
drawbacks of usual QKD. Usual security proofs of QDK make several
assumptions about the quantum systems, such as their Hilbert space
dimension. These assumptions are often critical: as we show below,
the security of the BB84 protocol, for instance, is entirely
compromised if Alice and Bob share four-dimensional systems
instead of sharing qubits as usually assumed. The problem is that
real-life implementation of QKD protocols may differ from the
ideal design. For instance, the quantum apparatuses may be noisy
or there may be uncontrolled side channels. A possible, but
challenging, way to address these problems would be to
characterize very precisely the quantum devices and try to adapt
the security proof to the actual implementation of the protocol.
The concept of device-independent QKD, on the other hand, applies
through its remarkable generality in a simple way to these
situations as it allows us to ignore all implementation details.

DIQKD makes it also easier to test the components of a QKD
protocol. Since its security relies on the observed classical data
generated by the devices, errors or deterioration with time of the
internal working of the quantum devices, which could be exploited
by an eavesdropper, are easily monitored and accounted for in the
key rate.

A third practical motivation for DIQKD is that it covers the
adverse scenario where the quantum devices are not trusted. For
instance, someone who had access to the quantum apparatuses at
some time might have hacked or modified their mechanism. But if
the devices still produce proper classical input-output relations,
which is all what is required, this is irrelevant to the security
of the scheme. To some extent DIQKD overturns the adage that the
security of a cryptographic system is only as good as its physical
security. Of course an eavesdropper who had access to the quantum
devices might have modified their working so that they directly
send her information about the measurement settings and outcomes.
But this goes against the basic requirement that Alice's and Bob's
locations should be completely secure against Eve's scrutiny -- a
necessary requirement for cryptography to have any meaning. It is
modulo this assumption, that the eavesdropper is
free to tamper with their devices.

\subsection{Usual QKD protocols are not secure in the device-independent scenario}
A consequence of adopting a more general security model is that
traditional QKD protocols may no longer be secure, as illustrated
by the following example.

Consider the entanglement-based version of BB84~\cite{BB84}. Alice
has a measuring device that takes a classical input $X\in\{0,1\}$
(her choice of measurement setting) and that produces an output
$a\in\{0,1\}$ (the measurement outcome). Similarly, Bob's device
accepts inputs $Y\in\{0,1\}$ and produce outputs $b\in\{0,1\}$.
Both measuring devices act on a two-dimensional subspace of the
incoming particles (e.g., the polarization of a photon). The
setting ``0'' is associated to the measurement of $\sigma_x$,
while the setting ``1'' corresponds to $\sigma_z$. Suppose that in
an ideal, noise-free situation they observe the following
correlations:
\begin{eqnarray}\label{bb84}
&&P(ab|00)=P(ab|11)=1/2\quad\text{if } a=b\nonumber\\
&&P(ab|01)=P(ab|10)=1/4\quad \text{for all }a,b\,,
\end{eqnarray}
where $P(ab|XY)$ is the probability to observe the pair of
outcomes $a,b$ given that they have made measurements $X,Y$. That
is, if Alice and Bob perform measurements in the same bases, they
always get perfectly correlated outcomes; while if they measure in
different bases, they get completely uncorrelated random outcomes.
In term of the measurement operators $\sigma_x$ and $\sigma_z$ and
the two-qubit state $|\psi\rangle\in
\mathbb{C}^2\otimes\mathbb{C}^2$ that characterizes their incoming
particles, the above correlations can be rewritten as
\begin{eqnarray}\label{bb84state}
\langle \psi |\sigma_x\otimes\sigma_x|\psi\rangle = \langle \psi |\sigma_z\otimes\sigma_z|\psi\rangle =1\nonumber\\
\langle \psi |\sigma_x\otimes\sigma_z|\psi\rangle =\langle \psi |\sigma_z\otimes\sigma_x|\psi\rangle =0\,.
\end{eqnarray}
The only state compatible with this set of equations is the
maximally entangled state $(|00\rangle+|11\rangle)/\sqrt 2$. Alice
and Bob therefore rightly conclude that they can safely extract a
secret key from their measurement data.

In the device-independent scenario, however, Alice and Bob can no
longer assume that the measurement settings ``0'' and ``1''
correspond to the operators $\sigma_x$ and $\sigma_z$, nor that
they act on the two-qubit space $\mathbb{C}^2\otimes\mathbb{C}^2$.
It is then not difficult to find separable (hence insecure) states
that reproduce the measurement data (\ref{bb84}) for appropriate
choice of measurements \cite{Magniez,AGM}. An example is given by
the $\mathbb{C}^4\otimes\mathbb{C}^4$ state
\[\label{bb84state2}
\rho_{AB}=\frac{1}{4}\sum_{z_0,z_1=0}^1
\left(|z_0\,z_1\rangle\langle
z_0\,z_1|\right)_A\otimes\left(|z_0\,z_1\rangle\langle
z_0\,z_1|\right)_B ,
\]
where the vectors $\ket{0}$ and $\ket{1}$ define the $z$ basis, and by the measurements
$\sigma_z\otimes I$ for the setting ``0'' and $I\otimes \sigma_z$
for the setting ``1''. Clearly this combination of state and
measurements reproduce the correlations (\ref{bb84}): Alice and
Bob find completely correlated outcomes when the use the same
measurement settings, and completely uncorrelated ones otherwise.
In contrast to the previous situation, however, Eve can now have a
perfect copy of the local states of Alice and Bob, for instance if
they share the tripartite state
\[
\rho_{ABE}=\frac{1}{4}\sum_{x,z=0}^1 \left(|z_0\,z_1\rangle\langle
z_0\,z_1|\right)_A\otimes\left(|z_0\,z_1\rangle\langle
z_0\,z_1|\right)_B\otimes \left(|z_0\,z_1\rangle\langle
z_0\,z_1|\right)_E\,.
\]
This example illustrates the fact that in the usual security
analysis of BB84 it is crucial to assume that Alice and Bob
measurements act on a two-dimensional space, a condition difficult
to check experimentally. If we relax this assumption, the security
is no longer guaranteed.

\subsection{How can DIQKD possibly be secure?}
Understanding better why usual QKD protocols are not secure in the
device-independent scenario may help us identify physical
principles on which to base the security of a device-independent
scheme. A first observation is that the correlations (\ref{bb84})
produced in BB84 are classical: we don't need to invoke quantum
physics at all to reproduce them. They can simply be generated by
a set of classical random data shared by Alice's and Bob's systems
--- in essence this is what the separable state (\ref{bb84state})
achieves. Formally, they can be written in the form
\[\label{classcorr}
P(ab|XY) = \sum_{\lambda} P(\lambda)\, D(a|X,\lambda)\,D(b|Y,\lambda)
\]
where $\lambda$ is a classical variable with probability
distribution $P(\lambda)$ shared by Alice's and Bob's devices and
$D(a|X,\lambda)$ is a function that completely specifies Alice's
outputs once the input $X$ and $\lambda$ are given (and similarly
for $D(b|Y,\lambda)$ ). An eavesdropper might of course have a
copy of $\lambda$, which would give her full information about
Alice's and Bob's outputs once the inputs are announced.

This trivial strategy is not available to the eavesdropper,
however, if the outputs of Alice's and Bob's apparatuses are
correlated in a non-local way, in the sense that they violate a
Bell inequality \cite{Bell}. Indeed, non-local correlations are
defined precisely as those that cannot be written in the form
(\ref{classcorr}). The violation of a Bell inequality is thus a
necessary requirement for the security of QKD protocol in the
device-independent scenario. This condition is clearly not
satisfied by BB84.

More than a necessary condition for security, non-locality is the
physical principle on which all device-independent security proofs
are based. This follows from the fact that non-local correlations
require for their generation entangled states, whose measurement
statistics cannot be known completely to an eavesdropper. To put
it in another way, Bell inequalities are the only entanglement
witnesses that are device-independent, in the sense that they do
not depend on the physical details underlying Alice's and Bob's
apparatuses.

\subsection{Earlier works and relation to QKD against no-signalling eavesdroppers}
The intuition that the security of a QKD scheme could be based on
the violation of a Bell inequality was at the origin of Ekert's
1991 celebrated proposal \cite{Ekert}. The crucial role that
non-local correlations play in a device-independent scenario was
also implicitly recognized by Mayers and Yao \cite{Mayers}.
Quantitative progress, however, has been possible only recently
thanks to the pioneering work of Barrett, Hardy, and Kent
\cite{BHK}. Barrett, Hardy, and Kent proved the security of QKD
scheme against general attacks by a supra-quantum eavesdropper
that is limited by the no-signalling principle only (rather than
the full quantum formalism). This is possible because once the
no-signaling condition is assumed, nonlocal correlations satisfy a
monogamy condition analogous to that of entanglement in quantum
theory \cite{nlinfo}. Since quantum theory satisfies the
no-signalling condition, security against a no-signalling
eavesdropper implies security in the device-independent scenario.

Barrett, Hardy, and Kent's result is a proof of principle as their
protocol requires Alice and Bob to have a noise-free quantum
channel and generates a single shared secret bit (but makes a
large number of uses of the channel). A slight modification of
their protocol based on the results of \cite{BKP} enables the
generation of a secret key of $\log_2 d$ bit if Alice and Bob have
a channel that distributes $d$-dimensional systems. Barrett,
Hardy, and Kent's work was extended to noisy situations and
non-vanishing key rates in \cite{AGM,AMP,Scarani}, though these
works only considered security against individual attacks, where
the eavesdropper is restricted to act independently on each of
Alice's and Bob's systems. Masanes {\sl et al.} introduced a
security proof valid against arbitrary attacks by an eavesdropper
that is not able to store non-classical information
\cite{masanes-winter}. This result was improved by Masanes
\cite{masanes} who proved security in the universally-composable
sense, the strongest notion of security. Although the last two
results take into account eavesdropping strategies that act
collectively on systems corresponding to different uses of the
devices, they require the no-signalling condition to hold not only
between the devices on Alice's and on Bob's side, but also between
all individual uses of the quantum device of one party. This
condition can be enforced, although not in a practical manner, by
having the parties use in parallel $N$ devices that are space-like
separated from each other, rather than using sequentially a single
device $N$ times.

There are fundamental motivations to study the security of QKD
protocols against no-signalling eavesdroppers (NSQKD); this
improves for instance our understanding of the relationship
between information theory and physical theories. From a practical
point of view, it is also interesting to develop cryptographic
schemes that rely on physical principles independent from quantum
theory and thus that could be guaranteed secure even if quantum
theory were ever to fail.

However, given that for the moment we have no good reasons (apart possibly theoretical ones) to doubt the validity of quantum theory, nor evidences that a hypothetical breakdown of quantum theory would signify the immediate end of quantum key distribution\footnote{For instance, quantum physics might only breakdown at an energy scale that would remain unaccessible to human control for ages.}, it is advantageous to exploit the full quantum formalism in the device-independent context. First of all, as the entire quantum formalism is more constraining than the no-signalling principle alone, we expect to derive higher key rates and better noise resistance in the quantum case (for instance, the proof of general security given in \cite{masanes} has a key rate and a noise-resistance that is not practical when applied to quantum correlations). A second advantage is that, from a technical point of view, we can exploit in proving security all the theoretical framework of quantum theory -- as opposed to a single principle. We may, in particular, use existing results such as de Finetti theorems, efficient privacy amplification schemes against quantum adversaries, etc. (but might also have to derive new technical results that may find applications in other contexts).

\subsection{Content and structure of the paper}
Here we prove the security of a modified version of the Ekert
protocol \cite{Ekert}, proposed in Ref. \cite{AMP}. Our proof,
already introduced in \cite{aci07}, exploits the full quantum
formalism, but is restricted to collective attacks, where Eve is
assumed to act independently and identically at each use of the
devices, though she can act coherently at any time on her own
systems. In the usual security model, security against collective
attacks implies security against the most general type of attacks
\cite{coll}. It is an open question whether this is also true in
the device-independent scenario. In the protocol that we analyze,
Alice and Bob bound Eve's information by estimating the violation
of the Clauser-Horne-Shimony-Holt (CHSH) inequality \cite{chsh}.
Our main result is a tight bound on the Holevo information between
Alice and Eve as a function of the amount of violation of the CHSH
inequality. The protocol that we use, our security assumptions,
and our main result are presented in Section~2. In particular, we
present in Subsection~2.4 all the details of our security proof,
which was only sketched in~\cite{aci07}.

It is crucial for the security of DIQKD that Alice's and Bob's
outcomes genuinely violate a Bell inequality. All experimental
tests of non-locality that have been made so far, however, are
subject to at least one of several loopholes and therefore admit
in principle a local description. We discuss in Section~3 the
issue of loopholes in Bell experiments from the perspective of~
DIQKD.

Finally, we conclude with a discussion of our results and some
open questions in Section~4.

\section{Results}
\subsection{The protocol}
The protocol that we study is a modification of the Ekert 1991
protocol \cite{Ekert} proposed in Ref.~\cite{AMP}. Alice and Bob
share a quantum channel consisting of a source that emits pairs of
particles in an entangled state $\rho_{AB}$. Alice can choose to
apply to her particle one out of three possible measurements
$A_0$, $A_1$ and $A_2$, and Bob one out of two measurements $B_1$
and $B_2$. All measurements have binary outcomes labeled by
$a_i,b_j\in\{+1,-1\}$.

The raw key is extracted from the pair $\{A_0,B_1\}$. The quantum
bit error rate (QBER) is defined as $Q=P(a\neq b|01)$.
This parameter estimates the amount of correlations between
Alice's and Bob's symbols and thus quantifies the amount of
classical communication needed for error correction. The
measurements $A_{1}$, $A_2$, $B_1$, and $B_2$ are used on a subset
of the particles to estimate the CHSH polynomial
\begin{equation}\label{CHSHeq}
    {S}=\moy{a_1b_1}+\moy{a_1b_2}+ \moy{a_2b_1}- \moy{a_2b_2}\,,
\end{equation}
where the correlator $\moy{a_ib_j}$ is defined as
$P\,(a=b|ij)-P\,(a\neq b|ij)$. The CHSH polynomial is used by
Alice and Bob to bound Eve's information and, thus, governs the
privacy amplification process. We note that there is no a priori
relation between the value of ${S}$ and the value of $Q$: these
are two parameters which are available to estimate Eve's
information.

Without loss of generality, we suppose that the marginals are
random for each measurement, i.e., $\moy{a_i}=\moy{b_j}=0$ for all
$i$ and $j$. Were this not the case, Alice and Bob could achieve
it a posteriori through public one-way communication by agreeing
on flipping randomly a chosen half of their bits. This
operation would not change the value of $Q$ and $S$ and would be
known to Eve.

A particular implementation of our protocol with qubits is given
for instance by the noisy two-qubit state $\rho_{AB}= p
\ket{\Phi^+}\bra{\Phi^+}+(1-p)I/4$ and by the qubit measurements
$A_0=B_1=\sigma_z$, $B_2=\sigma_x$, $A_1=(\si_z+\si_x)/{\sqrt{2}}$
and $A_2=(\si_z-\si_x)/\sqrt{2}$, which maximize the CHSH
polynomial for the state $\rho_{AB}$. The state $\rho_{AB}$ corresponds to a
two-qubit Werner state and arises, for instance, from the state
$\ket{\Phi^+}=1/\sqrt{2}(\ket{00}+\ket{11})$ after going through a
depolarizing channel, or through a phase-covariant cloner. The resulting correlations
satisfy $S=2\sqrt{2} p$ and $Q=1/2-p/2$, i.e.,
$S=2\sqrt{2}(1-2{Q})$. Though these correlations can be generated
in the way that we just described, it is important to stress that
Alice and Bob do not need to assume that they perform the above
measurements, nor that their quantum systems are of dimension 2,
when they bound Eve's information.

In the case of classically correlated data (corresponding to
$p\leq 1/\sqrt{2}$ for the above correlations), the maximum of the
CHSH polynomial (\ref{CHSHeq}) is 2, which defines the well-known
CHSH Bell inequality ${S}\leq 2$. Secure DIQKD is not possible if
the observed value of ${S}$ is below this classical limit, since
in this case there exists a trivial attack for Eve that gives her
complete information, as discussed in Subsection~1.3. On the other
hand, at the point of maximal quantum violation ${S}=2\sqrt 2$
(corresponding to $p=1$ for the above correlations), Eve's
information is zero. This follows from the work of
Tsirelson~\cite{cir80}, who showed that any quantum realization of
this violation is equivalent to the case where Alice and Bob
measure a two-qubit maximally entangled state. The main ingredient
in the security proof of our DIQKD protocol is a lower bound on
Eve's information as a function of the CHSH value. This bound
allows us to interpolate between the two extreme cases of zero and
maximal quantum violation and yields provable security for
sufficiently large violations.

\subsection{Eavesdropping strategies}
\subsection*{Most general attacks}
In the device-independent scenario, Eve is assumed not only to
control the source (as in usual entanglement-based QKD), but also
to have fabricated Alice's and Bob's measuring devices. The only
data available to Alice and Bob to bound Eve's knowledge is the
observed relation between the inputs and outputs, without
any assumption on the type of quantum measurements and systems
used for their generation.

In complete generality, we may describe this situation as follows.
Alice, Bob, and Eve share a state $\ket{\Psi}_{ABE}$ in
$H_A^{\otimes n}\otimes H_B^{\otimes n} \otimes H_E$, where $n$ is
the number of bits of the raw key. The dimension $d$ of Alice's and
Bob's Hilbert spaces $H_A=H_B=\mathbb{C}^d$ is unknown to them and
fixed by Eve. The measurement $M_k$ yielding the $k^\mathrm{th}$
outcome of Alice is defined on the $k^\mathrm{th}$ subspace of
Alice and chosen by Eve. This measurement may depend on the
input $A_{j_k}$ chosen by Alice at step $k$ and on the
value $c_k$ of a classical register stored in the device, that is,
$M_k=M_k(A_{j_k},c_k)$. The classical memory $c_k$ can in
particular store information about all previous inputs and
outputs. Note that the quantum device may also have a quantum memory,
but this quantum memory at step $k$ of the protocol can be seen as
part of Alice's state defined in $H_A^{k}$. The value of this
quantum memory can be passed internally from step $k$ of the
protocol to step $k+1$ by teleporting it from $H_A^k$ to
$H_A^{k+1}$ using the classical memory $c_k$. The situation is
similar for Bob.

\subsubsection*{Collective attacks}
In this paper, we focus on collective attacks where Eve applies
the same attack to each system of Alice and Bob. Specifically, we
assume that the total state shared by the three parties has the
product form $\ket{\Psi_{ABE}}=\ket{\psi_{ABE}}^{\otimes n}$ and
that the measurements are a function of the current input
only, e.g., for Alice $M_k=M(A_{j_k})$. We thus assume that the
devices are memoryless and behave identically and independently at
each step of the protocol. From now on, we simply write the
measurement $M(A_j)$ as $A_j$.

For collective attacks, the asymptotic secret key rate $r$ in the limit of a key of infinite size under one-way
classical postprocessing from Bob to Alice is lower-bounded by the
Devetak-Winter rate~\cite{DW}, \be \label{rate} r\, \geq\,
r_{DW}\,= \,I(A_0:B_1)\,-\, \chi(B_1:E)\,, \ee which is the
difference between the mutual information between Alice and Bob,
\be I(A_0:B_1)=H(A_0)+H(B_1)-H(A_0,B_1) \ee and the Holevo
quantity between Eve and Bob \be\chi(B_1:E)=S(\rho_E)-
\demi\sum_{b_1=\pm 1}S(\rho_{E|b_1})\,. \ee Here $H$ and $S$
denote the standard Shannon and von Neumann entropies,
$\rho_E=\linebreak[4]\text{Tr}_{AB}\ket{\psi_{ABE}}\bra{\psi_{ABE}}$
denotes Eve's quantum state after tracing out Alice and Bob's
particles, and $\rho_{E|b_1}$ is Eve's quantum state when Bob has
obtained the result $b_1$ for the measurement $B_1$. The optimal collective attack corresponds to the case where the tripartite state $\ket{\psi_{ABE}}$ is the purification of the bipartite state $\rho_{AB}$ shared by Alice and Bob.

 Since we have
assumed uniform marginals, the mutual information between Alice
and Bob is given here by \be I(A_0:B_1)=1-h(Q)\,, \ee where $h$ is the
binary entropy.

Note that the rate is given by \eqref{rate} and not by
$I(A_0:B_1)-\chi(A_0:E)$ because $\chi(A_0:E)\geq \chi(B_1:E)$
holds for our protocol \cite{AMP}; it is therefore advantageous
for Alice and Bob to do the classical postprocessing with public
communication from Bob to Alice.

\subsection{Security of our protocol against collective attacks}
To find Eve's optimal collective attack, we have to find the
largest value of $\chi(B_1:E)$ compatible with the observed
parameters $Q$ and $S$ without assuming anything about the
physical systems and the measurements that are performed. Our main
result is the following.

\begin{thm}
Let $\ket{\psi_{ABE}}$ be a quantum state and
$\{A_1,A_2,B_1,B_2\}$ a set of measurements yielding a violation
$S$ of the CHSH inequality. Then after Alice and Bob have
symmetrized their marginals, \be \chi(B_1:E)\leq
h\left(\frac{1+\sqrt{({S}/2)^2-1}}{2}\right)\,. \label{bound} \ee
\end{thm}
The proof of this Theorem will be given in
Subsection~\ref{sec:proof}. From this result, it immediately
follows that the key rate for given observed values of $Q$ and $S$
is \be\label{keyrate} r\geq
1-h(Q)-h\left(\frac{1+\sqrt{({S}/2)^2-1}}{2}\right)\,. \ee As an
illustration, we have plotted in Fig.~\ref{figcurves} the key rate
for the correlations introduced in Subsection~2.1 that satisfy
$S=2\sqrt{2}(1-2{Q})$ and which arise from the state
$\ket{\Phi^+}$ after going through a depolarizing channel.  We
stress that although with have specified a particular state and
particular qubit measurements that produce these correlations, we
do not assume anything about the implementation of the
correlations when computing the key rate. For the sake of
comparison, we have also plotted the key rate under the usual
assumptions of QKD for the same set of correlations. In this case,
Alice and Bob have a perfect control of their apparatuses, which
we have assumed to faithfully perform the qubit measurements given
in Subsection~2.1. The protocol is then equivalent to Ekert's,
which in turn is equivalent to the entanglement-based version of
BB84,
and one finds 
\be \chi(B_1:E)\leq h\left(Q+{S}/2\sqrt{2}\right)\,,
\label{boundstandard} \ee
as proven in Subsection~2.5.
If 
$S=2\sqrt{2}(1-2{Q})$, this expression yields the
well-known critical QBER of $11\%$ \cite{SP}, to be compared to
$7.1\%$ in the device-independent scenario (Fig.~\ref{figcurves}).
\begin{figure}\begin{center}
\includegraphics[scale=0.6]{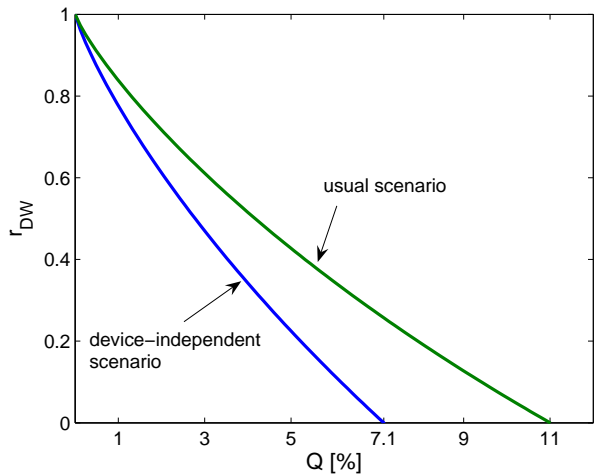}
\caption{Extractable secret-key rate against collective attacks in
the usual scenario [$\chi(B_1:E)$ given by
eq.~(\ref{boundstandard})] and in the device-independent scenario
[$\chi(B_1:E)$ given by eq.~(\ref{bound})], for correlations
satisfying $S=2\sqrt{2}(1-2{Q})$. The key rate is plotted as a
function of $Q$. Remember that the key rate for the BB84
protocol in the device-independent scenario is zero.
}\label{figcurves}
\end{center}\end{figure}

To illustrate further the difference between the
device-independent scenario and the usual scenario, we now give an
explicit attack which saturates our bound; this example also
clarifies why the bound (\ref{bound}) is independent of $Q$. To
produce correlations characterized by given values of $Q$ and $S$,
Eve sends to Alice and Bob the two-qubit Bell-diagonal state \be
\rho_{AB}({S})=\frac{1+{C}}{2}\,P_{\Phi^+}\,+\,
\frac{1-{C}}{2}\,P_{\Phi^-}\,, \label{rhoabc} \ee where
$P_{\Phi^\pm}$ are the projectors on the Bell states
$\ket{\Phi^\pm}=(\ket{00}\pm\ket{11})/\sqrt{2}$ and where
${C}=\sqrt{({S}/2)^2-1}$. She defines the measurements to be
$B_1=\si_z$, $B_2=\si_x$ and
$A_{1,2}=\frac{1}{\sqrt{1+{C}^2}}\si_z\pm\frac{{C}}{\sqrt{1+{C}^2}}\si_x$.
Any value of $Q$ can be obtained by choosing $A_0$ to be $\si_z$
with probability $1-2Q$ and to be a randomly chosen bit with
probability $2Q$. One can check that the Holevo information
$\chi(B_1:E)$ for the state (\ref{rhoabc}) and the measurement
$B_1=\si_z$ is equal to the righ-hand side of (\ref{bound}), i.e.,
this attack saturates our bound. This attack is impossible within
the usual assumptions because here not only the state $\rho_{AB}$,
but also the measurements taking place in Alice's apparatus depend
explicitly on the observed values of ${S}$ and $Q$. The state
(\ref{rhoabc}) has a nice interpretation: it is the two-qubit
state which gives the highest violation ${S}$ of the CHSH
inequality for a given value of the entanglement, measured by the
concurrence ${C}$ \cite{vw}. Therefore, for the optimal attack,
Eve uses the quantum state achieving the observed Bell violation
with the minimal amount of entanglement between Alice and Bob.
Since entanglement is a monogamous resource, this allows her to
maximize her correlations with the honest parties.

\subsection{Proof of upper bound on the Holevo quantity}\label{sec:proof}
The proof of the bound (\ref{bound}) was only sketched in
Ref.~\cite{aci07}. We present here all the details of that proof.
For clarity, we divide the proof in four steps.

\subsubsection{Step 1: Reduction to calculations on two qubits}

\begin{lemma} \label{lemmaone}
It is not restrictive to suppose that Eve sends to Alice and Bob a
mixture $\rho_{AB}=\sum_\lambda p_\lambda\,\rho_\lambda$ of
two-qubit states, together with a classical ancilla (known to her)
that carries the value $\lambda$ and determines which measurements
$A_{i}^\lambda$ and $B_{j}^\lambda$ are to be used on
$\rho_\lambda$.
\end{lemma}

The proof of this first statement relies critically on the
simplicity of the CHSH inequality (two binary settings on each
side). We present the argument for Alice, the same holds for Bob.
First, since any generalized measurement (POVM) can be viewed as a von Neumann measurement in a larger Hilbert space \cite{NielsenChuang}, we may assume that the two measurements $A_{1},A_{2}$ of Alice are von Neumann measurements, if necessary by including
ancillas in the state $\rho_{AB}$ shared by Alice and Bob. Thus
$A_1$ and $A_2$ are hermitian operators on $\mathbb{C}^d$ with
eigenvalues $\pm 1$. We can then use the following lemma.

\begin{lemma} \label{lemmaoneB}
Let $A_1$ and $A_2$ be Hermitian operators with eigenvalues equal
to $\pm 1$ acting on a Hilbert space $H$ of finite or countable
infinite dimension. Then we can decompose the Hilbert space $H$ as
a direct sum \be H= \oplus_\alpha H_\alpha^2 \ee such that
$\mathrm{dim}(H_\alpha^2)\leq 2$ for all $\alpha$, and such that
both $A_1$ and $A_2$ act within $H_\alpha^2$, that is, if
$|\psi\rangle \in H_\alpha^2$, then $A_1|\psi\rangle \in
H_\alpha^2$ and  $A_2|\psi\rangle \in H_\alpha^2$.
\end{lemma}

\begin{proof}
Previous proofs of this result have been obtained
independently by Tsirelson \cite{tsirelson} and
Masanes~\cite{masanes2}. Here, we provide an alternative and
possibly simpler proof.

Note that since the eigenvectors of $A_1$
and $A_2$ are $\pm 1$, these operators square to the identity:
$A_1^2 = A_2^2 = \one$. Therefore $A_2 A_1$ is a unitary operator.
Let $|\alpha\rangle$ be an eigenvector of $A_2 A_1$: \be A_2 A_1
|\alpha\rangle = \omega |\alpha\rangle \quad \mbox{with}\quad
|\omega|=1. \label{A2A1}\ee Then
$\ket{\tilde\alpha}=A_2\ket{\alpha}$ is also an eigenvector of
$A_2 A_1$ with eigenvalue $\overline\omega$, since $A_2
A_1\ket{\tilde\alpha}=A_2 A_1 A_2\ket{\alpha}\linebreak[1]=A_2
(A_2 A_1)^\dagger\ket{\alpha}=\overline \omega
A_2\ket{\psi}=\overline\omega\ket{\tilde\alpha}$. As $A_2 A_1$ is
unitary, its eigenvectors span the entire Hilbert space $H$. It
follows that $H$ can be decomposed as the direct sum
$H=\oplus_\alpha H_\alpha^2$, where
$H_c^2=\text{span}\{\ket{\alpha},\ket{\tilde\alpha}\}$ is (at
most) two-dimensional.

It remains to show that $A_1$ and $A_2$ act within $H_\alpha^2$.
By definition $A_2\ket{\alpha}=\ket{\tilde\alpha}$ and
$A_2\ket{\tilde\alpha}=\ket{\alpha}$. On the other hand,
$A_1\ket{\alpha}=A_1A_2\ket{\tilde\alpha}=\omega\ket{\tilde\alpha}$
and
$A_1\ket{\tilde\alpha}=A_1A_2\ket{\alpha}=\overline\omega\ket{\alpha}$.
Note that in the case where $\omega=\pm 1$, $A_1=\pm A_2$ on
$H_\alpha^2$, that is $A_1$ and $A_2$ are identical operators up
to a phase.
\end{proof}

\begin{proof}[Proof of Lemma \ref{lemmaone}]
We can rephrase Lemma \ref{lemmaoneB} as saying that
$A_j=\sum_\alpha P_\alpha A_j P_\alpha$ where the $P_\alpha$s are
orthogonal projectors of rank $1$ or $2$. From Alice's standpoint,
the measurement of $A_i$ thus amounts at projecting in one of the
(at most) two-dimensional subspaces defined by the projectors
$P_\alpha$, followed by a measurement of the reduced observable
$P_\alpha A_i P_\alpha=\vec{a}\,^\alpha_i\cdot\vec{\sigma}$.
Clearly, it cannot be worse for Eve to perform the projection
herself before sending the state to Alice and learn the value of
$\alpha$. The same holds for Bob. We conclude that without loss of
generality, in each run of the experiment Alice and Bob receive a
two-qubit state. The deviation from usual proofs of security of
QKD lies in the fact that the measurements to be applied can
depend explicitly on the state sent by Eve.
\end{proof}

\subsubsection{Step 2: Reduction to Bell-diagonal states of two qubits}
Let $\ket{\Phi^{\pm}}=1/\sqrt{2}\left(\ket{00}\pm\ket{11}\right)$
and $\ket{\Psi^{\pm}}=1/\sqrt{2}\left(\ket{01}\pm\ket{10}\right)$
be the four Bell states.
\begin{lemma} \label{lemmatwo}In the basis of Bell states ordered as $\{\ket{\Phi^+},\ket{\Psi^-}, \ket{\Phi^-},\ket{\Psi^+}\}$, each state $\rho_\lambda$ can be taken to
be a Bell-diagonal state of the form
\ba\rho_\lambda\left(\begin{array}{cccc}
\lambda_{\Phi^+}\\
& \lambda_{\Psi^-}\\
&& \lambda_{\Phi^-}\\
&&& \lambda_{\Psi^+}
\end{array}\right)\,,\label{belldiaglemma}\ea
with eigenvalues satisfying
\ba\lambda_{\Phi^+}\geq \lambda_{\Psi^-} &,&
 \lambda_{\Phi^-}\geq\lambda_{\Psi^+}\,. \label{orderlambdalemma}\ea
Furthermore, the measurements $A_i^\lambda$ and $B_j^\lambda$ can
be taken to be measurements in the $(x,z)$ plane.
\end{lemma}

\begin{proof}
For fixed $\lambda$ (we now omit the index $\lambda$), we can
label the axis of the Bloch sphere on Alice's side in such a way
that $\vec{a}_1$ and $\vec{a}_2$ define the $(x,z)$ plane; and
similarly on Bob's side.

Eve is a priori distributing any two-qubit state $\rho$ of which
she holds a purification. Now, recall that we have supposed,
without loss of generality, that all the marginals are uniformly
random. Knowing that Alice and Bob are going to symmetrize their
marginals, Eve does not loose anything in providing them a state
with the suitable symmetry. The reason is as follows. First note
that since the (classical) randomization protocol that ensures
$\moy{a_i}=\moy{b_j}=0$ is done by Alice and Bob through public
communication, we can as well assume that it is Eve who does it,
i.e., she flips the value of each outcome bit with probability one
half. But because the measurements of Alice and Bob are in the
$(x,z)$ plane, we can equivalently, i.e., without changing Eve's
information, view the classical flipping of the outcomes as the
quantum operation $\rho\rightarrow (\sigma_y\otimes\sigma_y)
\rho(\sigma_y\otimes\sigma_y)$ on the state $\rho$. We conclude
that it is not restrictive to assume that Eve is in fact sending
the mixture
\ba\bar{\rho}=\demi\left[\rho+(\sigma_y\otimes\sigma_y)
\rho(\sigma_y\otimes\sigma_y)\right]\,,\ea i.e., that she is
sending a state invariant under $\si_y\otimes\si_y$.

Now, $\ket{\Phi^+}$ and $\ket{\Psi^-}$ are eigenstates of
$\sigma_y\otimes\sigma_y$ for the eigenvalue $-1$, whereas
$\ket{\Phi^-}$ and $\ket{\Psi^+}$ are eigenstates of
$\sigma_y\otimes\sigma_y$ for the eigenvalue $+1$. Consequently,
$\bar{\rho}$ is obtained from $\rho$ by erasing all the coherences
between states with different eigenvalues. Explicitly, in the
basis of Bell states, ordered as $\{\ket{\Phi^+},\ket{\Psi^-},
\ket{\Phi^-},\ket{\Psi^+}\}$ we have \ba
\bar{\rho}&=&\left(\begin{array}{cccc}
\lambda_{\Phi^+} & r_1e^{i\phi_1}\\
r_1e^{-i\phi_1}& \lambda_{\Psi^-}\\
&& \lambda_{\Phi^-} & r_2e^{i\phi_2}\\
&& r_2e^{-i\phi_2}& \lambda_{\Psi^+}
\end{array}\right) \ea where all the non-zero elements coincide with those of the original $\rho$.

We now use some additional freedom that is left in the labeling:
we can select any two orthogonal axes in the $(x,z)$ plane to be
labeled $x$ and $z$, and we can also choose their orientation. We
make use of this freedom to bring $\bar{\rho}$ to the form \ba
\bar{\rho}&=&\left(\begin{array}{cccc}
\lambda_{\Phi^+} & ir_1\\
-ir_1& \lambda_{\Psi^-}\\
&& \lambda_{\Phi^-} & ir_2\\
&& -ir_2& \lambda_{\Psi^+}
\end{array}\right)\,,
\ea with $r_1$ and $r_2$ real and with the diagonal elements
arranged as \ba\lambda_{\Phi^+}\geq \lambda_{\Psi^-} &,&
 \lambda_{\Phi^-}\geq\lambda_{\Psi^+}\,. \label{orderlambda}\ea
 Indeed, let $R_y(\theta)=cos\frac{\theta}{2}\one + i
\sin\frac{\theta}{2}\si_y$: by applying $R_y(\alpha)\otimes
R_y(\beta)$  with \ba
\tan(\alpha-\beta)=\frac{2r_1\cos\phi_1}{\lambda_{\Phi^+}-\lambda_{\Psi^
-}}&,&
\tan(\alpha+\beta)=-\frac{2r_2\cos\phi_2}{\lambda_{\Phi^-}-\lambda_{\Psi^
+}}\,, \ea the off-diagonal elements become purely imaginary. In
order to further arrange the diagonal elements according to
(\ref{orderlambda}), one can make the following extra rotations:
\begin{itemize}
\item in order to relabel $\Phi^+\leftrightarrow \Psi^-$ without changing the others, one sets $\alpha-\beta=\pi$ and $\alpha+\beta=0$ i.e. $\alpha=-\beta=\frac{\pi}{2}$;
\item in order to relabel $\Phi^-\leftrightarrow \Psi^+$ without changing the others, one sets $\alpha-\beta=0$ and $\alpha+\beta=\pi$ i.e. $\alpha=\beta=\frac{\pi}{2}$;
\item in order to relabel both, one takes the sum of the previous ones i.e. $\alpha=\pi$ and $\beta=0$.
\end{itemize}
In this way one fixes $\lambda_{\Phi^+}\geq \lambda_{\Psi^-}$ and
$\lambda_{\Phi^-}\geq \lambda_{\Psi^+}$, i.e. the order of the
diagonal elements in each sector.


Finally, we repeat an argument similar to the one given above:
since $\bar{\rho}$ and its conjugate $\bar{\rho}^{\,*}$ produce
the same statistics for Alice and Bob's measurements and provide
Eve with the same information, we can suppose without loss of
generality that Alice and Bob rather receive the Bell-diagonal
mixture \ba \rho_\lambda&=&\demi
\left(\bar{\rho}+\bar{\rho}^{\,*}\right)\,=\,
\left(\begin{array}{cccc}
\lambda_{\Phi^+}\\
& \lambda_{\Psi^-}\\
&& \lambda_{\Phi^-}\\
&&& \lambda_{\Psi^+}
\end{array}\right)\,,\label{belldiag}\ea with the eigenvalues satisfying (\ref{orderlambda}).
\end{proof}

\subsubsection{Step 3: Explicit calculation of the bound}
\begin{lemma} \label{lemmathree}
For a Bell-diagonal state $\rho_{\lambda}$ (\ref{belldiaglemma})
with eigenvalues ${\lambda}$ ordered as in eq.
(\ref{orderlambdalemma}) and for measurements in the $(x,z)$
plane, \be \chi_{{\lambda}}(B_1:E)\leq
h\left(\frac{1+\sqrt{({S}_{{\lambda}}/2)^2-1}}{2}\right)\,,
\label{estimate2} \ee where
$S_{{\lambda}}$ 
 is the largest violation of the CHSH
inequality by the state $\rho_{\lambda}$.
\end{lemma}
In order to prove Lemma \ref{lemmathree} we have to bound
$\chi(B_1:E)=S(\rho_E)-\sum_{b_1=\pm
1}p(b_1)S\left(\rho_{E|b_1}\right)$. For the Bell diagonal state
(\ref{belldiaglemma}) one has \ba\label{chistep3}
\chi_{\lambda}(B_1:E)&=&H\left(\underline{\lambda}\right)\,-\,\demi
\left[S\left(\rho_{E|b_1=1}\right)+S\left(\rho_{E|b_1=-1}\right)\right]\,,
\ea where $H$ is Shannon entropy and where we have adopted the
notation $\underline{\lambda}\equiv\{ \lambda_{\Phi^+},
\lambda_{\Phi^-}, \lambda_{\Psi^+}, \lambda_{\Psi^-}\}$.
We divide the proof of Lemma \ref{lemmathree} into three parts. In
the first part, we prove that, for any given Bell-diagonal state,
Eve's best choice for Bob's measurement is $B_1=\sigma_z$, which
allows to express~(\ref{chistep3}) solely in term of the
eigenvalues $\underline{\lambda}$. In the second part, we obtain
an inequality between entropies. In the third part, we compute the
maximal violation of the CHSH inequality for states of the form
(\ref{belldiaglemma}).

\paragraph{Step 3, Part 1: Upper bound for given Bell-diagonal state}
\begin{lemma} \label{lemmafour}
For a Bell-diagonal state $\rho_{\lambda}$ with eigenvalues
${\lambda}$ ordered as in (\ref{orderlambdalemma}) and for
measurements in the $(x,z)$ plane, \be \chi_{{\lambda}}(B_1:E)\leq
H\left(\underline{\lambda}\right)-h(\lambda_{\Phi^+}+\lambda_{\Phi^-})\,.\label{chilambda2}
\ee
\end{lemma}
\begin{proof}
Let us compute $S\left(\rho_{E|b_1}\right)$. First, one gives Eve
the purification of $\rho_\lambda$: \ba \ket{\Psi}_{ABE}&=&
\sqrt{\lambda_{\Phi^+}}\ket{\Phi^+}\ket{e_1} +
\sqrt{\lambda_{\Phi^-}}\ket{\Phi^-}\ket{e_2} +
\sqrt{\lambda_{\Psi^+}}\ket{\Psi^+}\ket{e_3} +
\sqrt{\lambda_{\Psi^-}}\ket{\Psi^-}\ket{e_4} \ea with
$\braket{e_i}{e_j}=\delta_{ij}$. By tracing Alice out, one obtains
$\rho_{BE}$.

Now, Bob measures in the $x,z$ plane. His measurement $B_1$ can be
written as \ba
B_1=\cos{\varphi}\,\sigma_z+\sin\varphi\,\sigma_x\,.\ea After the
measurement, the system is projected in one of the eigenstates of
$B_1$ which can be written as
$\ket{b_1}=\sqrt{\frac{1+b_1\cos\varphi}{2}}\ket{0}+b_1\sqrt{\frac{1-b_1\cos\varphi}{2}}\ket{1}$
when $\varphi\in [0,\pi]$. The case $\varphi\in [\pi,2\pi]$
corresponds to a flip of the outcome $b_1$, but as the result that
follows is independent of the value of $b_1$, it is sufficient to
consider $\varphi\in [0,\pi]$.  The reduced density matrix of Eve
conditioned on the value of $b_1$ is given by \ba
\rho_{E|b_1}=\ket{\psi^+(b_1)}\bra{\psi^+(b_1)} +
\ket{\psi^-(-b_1}\bra{\psi^-(-b_1)}\,, \ea where we have defined
the two non-normalized states \ba
\ket{\psi^\sigma(b_1)}&=&\sqrt{\frac{1+b_1\cos\varphi}{2}}\left[
\sqrt{\lambda_{\Phi^+}}\ket{e_1} + \sigma
\sqrt{\lambda_{\Phi^-}}\ket{e_2} \right. \nonumber\\&& +  b_1
\sqrt{\frac{1-b_1\cos\varphi}{2}}\left[
\sqrt{\lambda_{\Psi^+}}\ket{e_3} +
\sigma\sqrt{\lambda_{\Psi^-}}\ket{e_4}\right]\,. \ea The
calculation of the eigenvalues of a rank two matrix is a standard
procedure. The result is that the eigenvalues of $\rho_{E|b_1}$
are independent of $b_1$ and are given by \ba
\Lambda_{\pm}=\demi\left(1\pm\,\sqrt{(\lambda_{\Phi^+}-\lambda_{\Psi^-})^2
+ (\lambda_{\Phi^-}-\lambda_{\Psi^+})^2+2\cos
2\varphi(\lambda_{\Phi^+}-\lambda_{\Psi^-})(\lambda_{\Phi^-}-\lambda_{\Psi^+})}\right)\,.
\ea Therefore we have obtained
$S\left(\rho_{E|b_1=1}\right)=S\left(\rho_{E|b_1=-1}\right)=h(\Lambda_+)$,
that is \ba
\chi_{\lambda}(B_1:E)=H\left(\underline{\lambda}\right)-h(\Lambda_+)\,.
\ea Now, for any set of $\lambda$'s, Eve's information is the
largest for the choice of $\varphi$ that minimizes $h(\Lambda_+)$,
which is the one for which the difference $\Lambda_+-\Lambda_-$ is
the largest. Because of (\ref{orderlambda}), the product
$(\lambda_{\Phi^+}-\lambda_{\Psi^-})(\lambda_{\Phi^-}-\lambda_{\Psi^+})$
is non-negative and the maximum is obtained for $\varphi=0$, i.e.,
$B_1=\sigma_z$. This gives the upper bound that we wanted \ba
\chi_{\lambda}(B_1:E)\leq
H\left(\underline{\lambda}\right)-h(\lambda_{\Phi^+}+\lambda_{\Phi^-})\,.\label{chilambda}
\ea
\end{proof}

\paragraph{Step 3, Part 2: Entropic Inequality}
\begin{lemma} \label{lemmafive}
Let $\underline{\lambda}$ be probabilities, i.e.
$\lambda_{\Phi^+}, \lambda_{\Phi^-}, \lambda_{\Psi^+},
\lambda_{\Psi^-} \geq 0$ and $\lambda_{\Phi^+}+ \lambda_{\Phi^-}+
\lambda_{\Psi^+}+ \lambda_{\Psi^-}=1$. Let $R^2 =
(\lambda_{\Phi^+}-\lambda_{\Psi^-})^2 +  (\lambda_{\Phi^-}-
\lambda_{\Psi^+})^2$. Then \ba
F(\underline{\lambda})=H\left(\underline{\lambda}\right)-h(\lambda_{\Phi^+}+\lambda_{\Phi^-})&\leq
&
 h\left(\frac{1+\sqrt{
 2 R^2
-1}}{2}\right) \quad \mbox{if $R^2 >1/2$}\label{estimate4a}\\
&\leq & 1 \quad \mbox{if $R^2 \leq 1/2$}
\,, \label{estimate4} \ea
with equality in eq. (\ref{estimate4a}) if and only if $\lambda_{\Phi^\pm}=0$ or $\lambda_{\Psi^\pm}=0$.
\end{lemma}

\begin{proof}
We can parameterize the $\lambda$'s as:
\ba
\lambda_{\Phi^+} &=& \frac{1}{4} + \frac{R}{2}\cos \theta + \delta \nonumber\\
\lambda_{\Phi^-} &=& \frac{1}{4} + \frac{R}{2}\sin \theta - \delta \nonumber\\
\lambda_{\Psi^-} &=& \frac{1}{4} - \frac{R}{2}\cos \theta + \delta \nonumber\\
\lambda_{\Psi^+} &=& \frac{1}{4} - \frac{R}{2}\sin \theta - \delta\,.
\ea
The  conditions
$\lambda_{\Phi^+}, \lambda_{\Phi^-}, \lambda_{\Psi^+}, \lambda_{\Psi^-} \geq 0$
 imply
\be
-\frac{1}{4} + \frac{R}{2}|\cos \theta | \leq \delta \leq \frac{1}{4} - \frac{R}{2}|\sin \theta|\,.
\label{domain1}
\ee
There is a solution for $\delta$ if and only if
\be\label{cond1}
|\cos \theta | + |\sin \theta|  \leq \frac{1}{R}\,.
\ee
This condition is non trivial if $R > 1/\sqrt{2}$.

When $R > 1/\sqrt{2}$, the extremal values of $\theta$, solution
of $|\cos \theta | + |\sin \theta|  = \frac{1}{R}$, correspond to
$\lambda_{\Phi^+}=0$ or $\lambda_{\Psi^-}=0$, and
$\lambda_{\Phi^-}=0$ or $\lambda_{\Psi^+}=0$. When both
$\lambda_{\Phi^\pm}=0$ or both $\lambda_{\Psi^\pm}=0$,
$F(\underline{\lambda})=h(1/2+(\sqrt{
 2 R^2-1})/2)$ and one has equality in eq.~(\ref{estimate4a}). In the other cases, $F(\underline{\lambda})=0$ and the inequality~(\ref{estimate4a}) is satisfied. Our strategy is to prove that when
$R > 1/\sqrt{2}$, the maximum of $F(\underline{\lambda})$ occurs
when  $|\cos \theta | + |\sin \theta|  = \frac{1}{R}$, i.e., at
the edge of the allowed domain for $\theta$. This will establish
(\ref{estimate4a}).

Let us start by finding the maximum of $F$ for fixed $R$ and
$\theta$. To this end, we compute the derivative of $F$ with
respect to $\delta$ \be \frac{\partial}{\partial \delta}
F\left(\underline{\lambda}\right) = - \log_2 \lambda_{\Phi^+}
+\log_2 \lambda_{\Phi^-} +\log_2 \lambda_{\Psi^+} -\log_2
\lambda_{\Psi^-}\,.\ee The derivative with respect to $\delta$
vanishes if and only if $ \lambda_{\Phi^+}  \lambda_{\Psi^-}  =
\lambda_{\Phi^-} \lambda_{\Psi^+}$, which is equivalent to \be
\delta = \delta^*(\theta)=\frac{R^2}{4}\left( \cos^2 \theta -
\sin^2 \theta \right)\,. \ee Note that $\delta^*(\theta)$ always
belongs to the domain (\ref{domain1}) for $\theta$ satisfying
(\ref{cond1}), i.e., it is an extremum of $F$. We also have that
\be \frac{\partial^2}{\partial_\delta^2}
F\left(\underline{\lambda}\right) = - \frac{1}{\lambda_{\Phi^+}} -
\frac{1}{ \lambda_{\Phi^-} } - \frac{1}{\lambda_{\Psi^+} } -
\frac{1}{\lambda_{\Psi^-}} < 0\,,\ee which shows that
$\delta^*(\theta)$ is a maximum of $F$ (not a minimum).

We have thus identified the unique maximum of $F$ at fixed
$\theta$. Let us now take the optimal value of $\delta =
\delta^*(\theta)$, and let $\theta$ vary. We compute the
derivative of $F$ with respect to $\theta$ along the curve $\delta
= \delta^*(\theta)$: \ba \frac{d}{d \theta} F |_{\delta=\delta^*}
&=& \frac{\partial}{\partial_\delta} F |_{\delta=\delta^*}\frac{d
\delta^*(\theta)}{d \theta} + \frac{\partial}{\partial_\theta} F
|_{\delta=\delta^*} = \frac{\partial}{\partial_\theta}
F |_{\delta=\delta^*}\nonumber\\
&=& -\frac{R}{2} \cos \theta \log_2 \left( \frac{ \lambda_{\Phi^-}
(\lambda_{\Psi^+} + \lambda_{\Psi^-} ) } {\lambda_{\Psi^+} (
\lambda_{\Phi^+} + \lambda_{\Phi^-})}\right) +\frac{R}{2} \sin
\theta \log_2 \left( \frac{  \lambda_{\Phi^+} (\lambda_{\Psi^+} +
\lambda_{\Psi^-} ) } {\lambda_{\Psi^-} ( \lambda_{\Phi^+} +
\lambda_{\Phi^-})}\right)\,. \ea Now, when $\delta=\delta^*$, we
have the identities: \be \frac{ \lambda_{\Phi^+}  } {
\lambda_{\Phi^+} + \lambda_{\Phi^-}} =\frac{ \lambda_{\Psi^+}  } {
\lambda_{\Psi^+} + \lambda_{\Psi^-}} = \frac{1}{2} + \frac{R}{2}
\cos \theta - \frac{R}{2} \sin \theta\,, \ee and \be \frac{
\lambda_{\Phi^-}  } { \lambda_{\Phi^+} + \lambda_{\Phi^-}} =\frac{
\lambda_{\Psi^-}  } { \lambda_{\Psi^+} + \lambda_{\Psi^-}} =
\frac{1}{2} - \frac{R}{2} \cos \theta + \frac{R}{2} \sin \theta\,.
\ee Using these relations we obtain \ba \frac{d}{d_\theta} F
|_{\delta=\delta^*} &=& -\frac{R}{2} \left ( \cos \theta + \sin
\theta  \right) \log_2 \frac{ 1 - R\cos \theta + R \sin \theta } {
1 + R\cos \theta - R \sin \theta }\,. \ea This quantity vanishes
(i.e. we have an extremum) if and only if $\cos \theta + \sin
\theta=0$ or $\cos \theta - \sin \theta =0$, that is $\theta=\pm
\pi/4, \pm 3 \pi/4$.

When $R> 1/\sqrt{2}$ the points $\theta=\pm \pi/4, \pm 3 \pi/4$
lie outside the allowed domain for $\theta$. Hence the maximum of
$G$ occurs when $\theta$ lies at the edge of its allowed domain.
As discussed above, this proves our claim when $R> 1/\sqrt{2}$.

When $R\leq 1/\sqrt{2}$, the extrema can be reached. Note that
$\theta=\pm \pi/4, \pm 3 \pi/4$ implies $\delta^*=0$. One then
easily checks that the maximum of $F$ occurs when $\theta = \pi/4,
-3 \pi/4$, whereupon $F=1$.  This establishes
eq.~(\ref{estimate4}).
\end{proof}

\paragraph{Step 3, Part 3: Violation of CHSH}
\begin{lemma}\label{lemmaseven}
The maximal violation $S_\lambda$ of the CHSH inequality for a
Bell diagonal state $\rho_\lambda$ given by (\ref{belldiaglemma})
with eigenvalues ordered according to  (\ref{orderlambdalemma}) is
\be\label{chshmax} S_\lambda = \max\left\{ 2\sqrt{2}\,\sqrt{
(\lambda_{\Phi^+}- \lambda_{\Psi^-})^2 + (\lambda_{\Phi^-} -
\lambda_{\Psi^+})^2 }\,,\,2\sqrt{2}\,\sqrt{(\lambda_{\Phi^+}-
\lambda_{\Psi^+})^2 + (\lambda_{\Phi^-} - \lambda_{\Psi^-})^2}
\right\} \ee
\end{lemma}
\begin{proof}
For any given two-qubit state $\rho$, the maximum value of the
CHSH expression can be computed using the following recipe
\cite{horodecki}: let $T$ be the tensor with entries
$t_{ij}=\mathrm{Tr}\left[\sigma_i\otimes\sigma_j\,\rho\right]$,
and let $\tau_1$ and $\tau_2$ be the two largest eigenvalues of
the symmetric matrix $T^TT$. Then, for optimal measurement
$S=2\sqrt{\tau_1+\tau_2}$.

We are working with the Bell-diagonal state (\ref{belldiaglemma}),
for which \ba T_\lambda=\left(\begin{array}{ccc} \lambda_{\Phi^+}-
\lambda_{\Phi^-}+ \lambda_{\Psi^+}- \lambda_{\Psi^-}\\ &
-\lambda_{\Phi^+}+\lambda_{\Phi^-}+\lambda_{\Psi^+}-\lambda_{\Psi^-}\\
&&\lambda_{\Phi^+}+\lambda_{\Phi^-}-\lambda_{\Psi^+}-\lambda_{\Psi^-}\end{array}\right).
\ea
Taking into account the order (\ref{orderlambdalemma}), one has
$T_{zz}\geq |T_{xx}|$.
Hence either
\be \tau_1 +\tau_2 = T_{zz}^2 + T_{xx}^2 =
2\left[(\lambda_{\Phi^+}- \lambda_{\Psi^-})^2 +
(\lambda_{\Phi^-} - \lambda_{\Psi^+})^2\right]\,,\ee
 or
\be \tau_1 +\tau_2 = T_{zz}^2 + T_{yy}^2 =
2\left[(\lambda_{\Phi^+}-  \lambda_{\Psi^+})^2 +
(\lambda_{\Phi^-} - \lambda_{\Psi^-})^2\right]\,.\ee
\end{proof}

We can now provide the proof of lemma \ref{lemmathree}.
\begin{proof}[Proof of lemma \ref{lemmathree}]
In the case that $S_\lambda$ is equal to the first expression in
(\ref{chshmax}), Lemma \ref{lemmathree} immediately follows from
combining Lemmas \ref{lemmafour} and \ref{lemmafive}, since
$S_\lambda=2\sqrt{2}R$. Note that the threshold $R^2=1/2$ in Lemma
\ref{lemmafive} corresponds to the threshold for violating the
CHSH inequality.

In the other case, we once again combine Lemmas \ref{lemmafour}
and \ref{lemmafive}, and note that the function $F$ in Lemma
\ref{lemmafour} is invariant under permutation of
$\lambda_{\Psi^+}$ and $\lambda_{\Psi^-}$ with $\lambda_{\Phi^+}$,
$\lambda_{\Phi^-}$ fixed.
\end{proof}

%

\subsubsection{Step 4: Convexity argument}
To conclude the proof of the Theorem, note that if Eve sends a
mixture of Bell-diagonal states $\sum_{\lambda}
p_\lambda\,\rho_{\lambda}$ and chooses the measurements to be in
the $(x,z)$ plane, then $\chi(B_1:E)=\sum_\lambda p_\lambda\,
\chi_{\lambda}(B_1:E)$. Using \eqref{estimate2}, we then find
$\chi(B_1:E)\leq \sum_\lambda p_\lambda\, F(S_\lambda)\leq
F(\sum_\lambda p_\lambda\, S_\lambda)$, where the last inequality
holds because $F$ is concave. But since the observed violation $S$
of CHSH is necessarily such that $S\leq \sum_\lambda p_\lambda
S_\lambda$ and since $F$ is a monotonically decreasing function,
we find $\chi(B_1:E)\leq F(S)$.

\subsection{Derivation of the bound (\ref{boundstandard}) in the standard scenario}
In the standard scenario, Alice and Bob know that they are measuring qubits and have set their measurement settings in the best possible way for the reference state $\ket{\Phi^+}$. We assume one such possible choice (all the others being equivalent), the one specified in Subsection~2.1: $A_0=B_1=\sigma_z$, $A_1=(\sigma_z+\sigma_x)\sqrt{2}$, $A_2=(\sigma_z-\sigma_x)\sqrt{2}$, $B_2=\sigma_x$. Thus the CHSH polynomial becomes
\ba
CHSH&=&\sqrt{2}\left(\sigma_x\otimes \sigma_x+\sigma_z\otimes \sigma_z\right)\,.\label{formchsh}
\ea
The calculation of the unconditional security bound follows exactly the usual one, as presented for instance in Appendix A of Ref.~\cite{revqkd}. As well-known, in the usual BB84 protocol, the measured parameters are the error rate in the $Z$ and in the $X$ basis, $\varepsilon_{z,x}$; if the $Z$ basis is used for the key and the $X$ basis for parameter estimation, Eve's information is bounded by \ba\chi(A_Z:E)=\chi(B_Z:E)=h(\varepsilon_x)\label{boundusual}\,.\ea

In our case, $\varepsilon_z=Q$; but instead of $\varepsilon_x$, the parameter from which Eve's information is inferred is the average value $S$ of the CHSH polynomial. Given (\ref{formchsh}), the evaluation of $S$ on a Bell-diagonal state is straightforward: $S=2\sqrt{2}(\lambda_1-\lambda_4)$. Now, with the parametrization $\lambda_1=(1-\varepsilon_z)(1-u)$ and $\lambda_4=\varepsilon_z v$, we immediately obtain $\lambda_1-\lambda_4=1-\varepsilon_z-[(1-\varepsilon_z)u+\varepsilon_z v]= 1-\varepsilon_z-\varepsilon_x$ because of Eq.~(A7) of Ref.~\cite{revqkd}. Therefore $S=2\sqrt{2}(1-Q-\varepsilon_x)$ i.e. \ba
\varepsilon_x=1-Q-S/2\sqrt{2}\,.\ea Since $h(\varepsilon_x)=h(1-\varepsilon_x)$, this leads immediately to (\ref{boundstandard}).

\section{Loopholes in Bell experiments and DIQKD}\label{sec:loopholes}
The security of our protocol, like the security of any DIQKD
protocol, relies on the violation of a Bell inequality. All
experimental tests of Bell inequalities that have been made so
far, however, are subject to at least one of several loopholes and
therefore admit in principle a local description.
We discuss here how these loopholes can impact DIQKD protocols.

\subsection{Loopholes in Bell experiments}
Basically, a loophole-free Bell experiment requires two
ingredients: i) no information about the input of one party should
be known to the other party before she has produced her output;
ii) high enough detection efficiencies.

If the first requirement is not fulfilled, the premises of Bell's
theorem are not satisfied and it is trivial for a classical model
to account for the apparent non-locality of the observed
correlations. In practice, this means that the measurements should
be carried out sufficiently fast and far apart from each other so
that no sub-luminal influence can propagate from the choice of
measurement on one wing to the measurement outcome on the other
wing. Additionally, the local choices of measurement
should not be determined in advance,
i.e., they should be truly random events. Failure to satisfy one
of these two conditions is known as the locality
loophole~\cite{Bell}.

The second requirement arises from the fact that in practice not
all signals are detected by the measuring devices, either
because of inefficiencies in the devices themselves, or because of
particle losses on the path from the source to the detectors. The
detection loophole~\cite{pearle} exploits the idea that it is a
local variable that determines whether a signal will be registered
or not. The particle is detected only if the setting of the
measuring device is in agreement with a predetermined scheme. In
this way, apparently non-local correlations can be reproduced by a
purely local model provided that the efficiency $\eta$ of the
detectors is below a certain threshold. In general the efficiency
necessary to rule out a local description depends on the Bell
inequality that is tested, and is quite high for Bell inequalities
with low numbers of inputs and outputs (for the CHSH inequality,
one must have $\eta> 82.8\%$). It is an open question whether
there exist Bell inequalities (with reasonably many inputs and
outputs) allowing significantly lower detection efficiencies (see
e.g.~\cite{zoology,I4422,Vertesi2}). From the point of view of the
data analysis, to decide whether an experiment with inefficient
detector has produced a genuine violation of a Bell inequality,
all measurement events, including no-detection events, should be
taken into account in the non-locality test.

All Bell experiments performed so far suffer from (at least) one
of the two above loopholes. On the one hand, photonic experiments
can close the locality loophole~\cite{aspect,weihs,tittel}, but
cannot reach the desired detection efficiencies. On the other
hand, experiments carried out on entangled
ions~\cite{rowe,matsukevich} manage to close the detection
loophole, but are unsatisfactory from the point of view of
locality. Note that other loopholes or variants of the above
loophole have also been identified, such as the coincidence-time
loophole~\cite{gill}, but these are not as problematic.

\subsection{Loopholes from the perspective of DIQKD}
When considering the implications of these loopholes for DIQKD, a
first point to realize is that they are mainly a technological
problem, but do not in any way undermine the concept of DIQKD
itself. An eavesdropper trying to exploit one of the above
loopholes would clearly have to tamper with Alice and Bob's
devices, but it is not necessary for Alice and Bob to ``trust'' or
characterize the inner working of their devices to be sure that
all loopholes are closed. This can be decided solely by looking at
the classical input-output relations produced by the quantum
devices (and possibly their timing). In other words, we do not
have to leave the paradigm of DIQKD to guarantee the security of
the protocol (though of course with present-day technology it
might be difficult to construct  devices that pass such security
tests).

A second important observation is that there is a fundamental
difference between a Bell experiment whose aim is to establish the
non-local character of Nature and a quantum key distribution
scheme based on the violation of a Bell inequality. In the first
case, we are trying to rule out a whole set of models of Nature
(including models that can overcome the laws of physics as they
are currently known), while in the second case, we are merely
fighting an eavesdropper limited by the laws of quantum physics.

Seen in this light, the locality loophole is not problematic in
our context. In usual Bell experiments, the locality loophole is
dealt with by enforcing a space-like separation between Alice and
Bob. This guarantee that no sub-luminal signals (including signals
mediated by some yet-unknown theory) could have traveled between
Alice's and Bob's devices. In the context of DIQKD, it is
sufficient to guarantee that no quantum signals (e.g. no photon)
can travel from Alice to Bob. This can be enforced by a proper
isolation of Alice's and Bob's locations. As stated in
Section~1.1, we make here the basic assumption, shared by usual
QKD and without which cryptography wouldn't make any sense, that
Alice's and Bob's locations are secure, in the sense that no
unwanted information can leak out to the outside. Whether this
condition is fulfilled is an important question in practice, but
it is totally alien to quantum key distribution, whose aim is to
establish a secret key between two parties given that this
assumption is satisfied. In a similar way, we assume here, as in
usual QKD, that Alice and Bob choose their measuring settings with
trusted random number generators whose outputs are unknown to Eve.
The locality loophole is thus not a fundamental loophole in the
context of DIQKD and can be dealt with using today's technology.

The detection loophole, on the other hand, is a much more
complicated issue. Experimental tests of non-locality circumvent
this problem by discarding no-detection events and recording only
the events where both measuring devices have produced an answer.
This amounts to perform a post-selection on the measurement data.
This post-selection is usually justified by the fair sampling
assumption, which says that the sample of detected particles is a
fair sample of the set of all particles, i.e., that there are no
correlations between the state of the particles and their
detection probability. Although it may be very reasonable to
expect such a condition to hold for any realistic model of Nature,
it is clearly unjustified in the context of DIQKD, where we assume
that the quantum devices are provided by an untrusted party
\cite{larsson,lo}. In our context, it is thus crucial to close the
detection loophole. This has already been done for some
experiments \cite{rowe,matsukevich} although not yet on distances
relevant for QKD.

Note that a proper security analysis of DIQKD with inefficient
detectors has to take into account all measurement outcomes
produced by the devices, which in our case would include the
outcomes ``1'', ``-1'', and the no-detection outcome ``$\bot$''. A
possible strategy to apply our proof to this new situation simply
consists in viewing the absence of a click ``$\bot$'' as a ``-1''
outcome, thus replacing a 3-output device by an effective 2-output
device. To give an idea of the amount of detection inefficiency
that can be tolerated in this way, we have plotted in
Figure~\ref{figdet} the key rate as a function of the efficiency
of the detectors for the ideal set of quantum correlations
that give the maximal violation of the CHSH inequality, obtained
when measuring a $\ket{\phi^+}$ state. The key rate is given by
Eq.~\eqref{keyrate} with $Q=\eta(1-\eta)$ and $S= 2\sqrt
2\eta^2+2(1-\eta)^2$.
\begin{figure}\begin{center}
\includegraphics[scale=0.6]{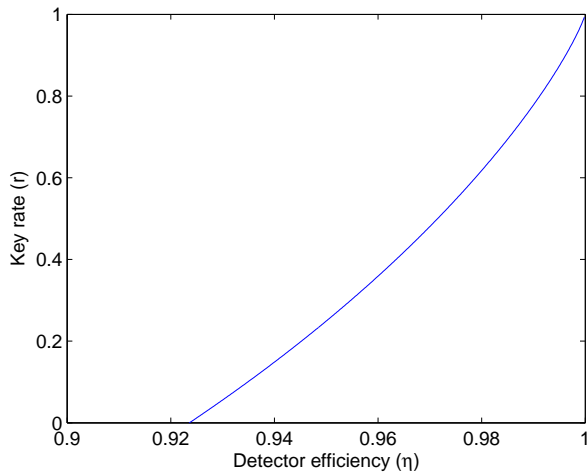}
\caption{Key rate as a function of detection efficiency for the ideal correlations coming from the maximally entangled state $\ket{\phi^+}$ and satisfying $Q=0$ and $S=2\sqrt{2}$, obtained by replacing the absence of a click by the outcome $-1$. The efficiency threshold above which a positive key rate can be extracted is $\eta=0.924$. \label{figdet}}
\end{center}\end{figure}

\subsection{Ideas for overcoming the detection loophole}

As mentioned above, the experiment of \cite{matsukevich}, which is based on
entanglement swapping between two ions separated by about 1 meter, is immune to the detection loophole. A natural way to implement our DIQKD protocol would thus be to improve this experiment. This would require extending the distance between the ions, but also improving the visibility and significantly
improving the data rate (currently one event every 39 seconds). This approach could of course in principle also be implemented with neutral atoms, quantum dots, etc. Here we discuss ideas on how
the problem of the detection loophole could be solved (at least partially) within an all photonic implementation, using heralded quantum memories and trusted detectors.

In a realistic quantum key distribution scenario there are
basically two kinds of losses that should be studied separately:
line losses and detector losses. Line losses are due (in practice)
to the imperfections of the quantum channel between Alice and Bob.
However, as far as the theoretical security analysis is concerned,
these losses should be assumed to be the result of Eve's actions,
since Alice and Bob do not control the quantum channel. One
possibility for Alice and Bob to overcome this problem is to use
heralded quantum memories. Using this technique, Alice and Bob can
know whether their respective memory device is loaded or not, that
is whether a photon really arrived in their device or not. In the
case that both memories are loaded, they release the photons and
perform their measurements. This procedure thus implements a kind
of quantum non-demolition measurement of the incoming states,
which allows Alice and Bob to get rid of the losses of the quantum
channel. This should be realizable within a few years, thanks to
the development of quantum repeaters \cite{TittelRev,Qinternet}.

The second type of losses, the detector losses, are probably more
crucial. We can, however, consider the situation where Alice's and
Bob's detectors are not part of the uncharacterized quantum
devices. That is, the quantum devices of Alice and Bob are viewed
as black-boxes that receive some classical input and produce an
output signal which is later detected, and transformed into the
final classical outcome, by a separate detector\footnote{Note, we
are not considering here the ``detector'' as a complete
measurement device, but only as the part of the device that clicks
or not whenever it is hit by one or several photons. In
particular, all the machinery that determines the choice of
measurement bases (and which may include beam-splitters,
polarizers, etc.) is still assumed to be part of the black-box
device.}. The detectors may be assumed to be trusted by and under
the control of Alice and Bob or they can be tested independently
from the rest of the quantum devices. Alice and Bob can for
instance do a tomography of their quantum detectors
\cite{tomo_det99,tomo_det01,tomo_det04,tomo_det08}, which consists
in determining the measurement that these detectors
actually perform. Such detector tomography, which has been
recently demonstrated experimentally in \cite{tomo_det08}, clearly
limits Eve's ability to exploit the detection loophole. This kind
of analysis may require to elaborate counter-measures against some
sort of a trojan-horse attacks on the detectors, in which Alice's
device (manufactured by Eve), sometimes sends nothing and
sometimes sends bright pulses in order to ensure that a detection
occurs. We believe that the power of such attacks can be severely
constrained by placing multiple detectors instead of one at each
output mode of Alice's and Bob's devices \cite{trojan}.

In the scenario that we  outlined in the preceding paragraph,
we have made a move to a situation that is intermediate between
usual QKD, where all devices are assumed to be trusted, and DIQKD,
where all quantum devices are untrusted. In this new situation,
Alice and Bob either need to trust their detectors (in the same
way as that they trust their random number generators or the
classical devices) or they need to test them with a trusted
calibration device (that they should get from a different provider
than Eve).  Whether this is a reasonable or practical scenario to
consider depends on the respective difficulty of testing the
detectors vs the entire quantum devices, and on the advantages
that may follow from trusting part of the quantum devices (this
may still allow for instance to forget about side-channels, or
imperfections in the measurement bases).

\section{Discussion and open questions}

Identifying the minimal set of physical assumptions allowing
secure key distribution is a fascinating problem, both from a
fundamental and an applied point of view. DIQKD possibly
represents the ultimate step in this direction, since its security
relies only on a fundamental set of basic assumptions: (i) access
to secure locations, (ii) trusted randomness, (iii)
trusted classical processing devices, (iv) an authenticated
classical channel, (v) and finally the general validity of a
physical theory, quantum theory. In this work, we have shown that
for the restricted scenario of collective attacks, secure DIQDK is
indeed possible. There remain, however, plenty of interesting open
questions in the device-independent scenario.

From an applied point of view, the most relevant questions are
related to loopholes in Bell tests, particularly the detection
loophole, as discussed in Section~\ref{sec:loopholes}. The detection loophole, usually seen mostly as a foundational problem, thus becomes a relevant issue from an applied perspective, with important implications for cryptography \footnote{Recently, the role of the detector efficiency loophole in standard QKD has been analyzed in Ref.~\cite{lutkenhaus}}. From a
theoretical point of view, it would be highly desirable to extend
the security proof presented here to other scenarios, the ultimate
goal being a general security proof. We list below several
possible directions to extend our results.
\begin{itemize}
    \item As we discussed, the violation of a Bell inequality represents a
    necessary condition for secure DIQKD. It would be interesting, then, to
    consider other protocols, based on different Bell inequalities,
    even under the additional
    assumption of collective attacks. Some interesting questions are:
    (i) how does the security of DIQKD change when using larger
    alphabets, especially when compared with
    standard QKD~\cite{qudit1,qudit2,errfiltr}?
    (ii) can one establish
    more general relations bounding Eve's information from the amount
    of observed Bell inequality violation?
    \item A key ingredient in our security
    proof is the fact that it is possible to reduce the
    whole analysis to a two-qubit optimization problem. This is because
    any pair of quantum binary measurements can be decomposed as the direct sum of pairs of measurements
    acting on two-dimensional spaces. Do similar results exist for
    more complex scenarios? More generally, are all possible bipartite quantum
    correlations for $m$ measurements of $n$ outcomes, for finite
    $m$ and $n$,
    attainable by measuring finite-dimensional
    systems~\cite{qcorr}? Some progress on this question was
    recently obtained in Refs.~\cite{vertpal,briet}, where it was shown that infinite
    dimensional systems are needed to generate all
    two-outcome ($n=2$) quantum correlations, thus proving a conjecture made in Ref. \cite{dimH}. The proof of this result,
    however, is only valid when $m\rightarrow\infty$.
        \item Our security analysis works for the case of
    one-way reconciliation protocols. How is the security of the protocol
    modified when two-way reconciliation techniques are
    considered? Does then a Bell inequality violation represent a sufficient condition for
    security? In this direction, it was shown in Ref.~\cite{MAG} that
    all correlations violating a Bell inequality contain some
    form of secrecy, although not necessarily distillable into a
    key.
    \item In the spirit of removing the largest number of assumptions necessary for the security of QKD, an interesting extension has been anticipated by Kofler, Paterek and Brukner \cite{kpb}. They noticed that quantum cryptography may be secure even when one allows the eavesdropper to have partial information about the measurement settings. To illustrate this scenario on our protocol, we suppose that in each run Eve has some probability to make a correct guess on the choice of measurement settings. The best way to model this situation from the perspective of Eve is to have an additional bit $f$ (``flag'') such that $f = 1$ guarantees her guess to be correct, while $f = 0$ implies that her guess is uncorrelated with the real settings: indeed, any scenario with partial knowledge may be obtained by Eve forgetting the value of $f$. We suppose that the case $f=1$ happens with probability $q$ and $f=0$ with probability $1-q$. When Eve has full information on Bob's measurement choice, she can fix in advance Alice's and Bob's outcome while at the same time engineering a violation of CHSH up to the algebraic limit of 4. If Eve follows this strategy, the observed violation will then be $S=4q+(1-q)S'$ and the security bound will be given by \ba\chi(B_1:E)&=&q+(1-q)h\left([1+\sqrt{(S'/2)^2-1}] /2\right)\,.\ea
        This proves that there cannot be any security if $q\geq \sqrt{2}-1\approx 41\%$. It would be interesting to consider more elaborated situations, e.g. those in which Eve may have partial information about sequences of measurement settings.
    \item In standard QKD, it is known that security against collective attacks implies security against the most general type of attacks. This follows from an application of the exponential quantum De
    Finetti theorem of Ref.~\cite{renner}, but can also be proven trough a direct argument \cite{coll}.
    Does a similar result
    hold in the device-independent scenario? In particular can the exponential de Finetti theorem \cite{renner} be extended to the device-independent scenario? If this was the case, our security proof would automatically be promoted to a general security proof.
    Some preliminary results in this direction have
    been obtained in Refs.~\cite{df1,barrett09}, where two different versions of
    a de Finetti theorem for general
    no-signaling probability distributions were derived.

    Or could it be that collective attacks are strictly weaker than general attacks in the device-independent scenario? Here the main difficulty for deriving a general security proof is that, contrary to standard QKD, the devices may behave in a way that depends on previous inputs and outputs. In particular, the measurement setting could be different in each round and depend on the results of previous measurements. It is not clear what role such memory effects play in the device-independent scenario, and whether it would be possible to find an explicit attack exploiting them which would outperform any collective (hence memoryless) attack.
    \item A final possibility would be to adapt the techniques developed
    in Refs.\cite{masanes-winter,masanes}, valid for the general
    case of no-signaling correlations, to the quantum
    scenario. The results of these works prove the security of QKD protocols against
    eavesdroppers limited only by the no-signaling principle. Unfortunately, the corresponding key rates and noise-resistance are at present unpractical when applied to correlations that can be obtained by measuring quantum states.
    A natural question is then: how can one incorporate the constraints associated to the
    quantum formalism to these techniques in order to obtain
    better key rates and better noise-resistance for quantum correlations?

\end{itemize}

\section*{Acknowledgements}\addcontentsline{toc}{section}{Acknowledgements}
 We are grateful to C.~Branciard,
I.~Cirac, A.~Ekert, A.~Kent, Ll. Masanes, T.~Paterek, R.~Renner and C.~Simon for fruitful
discussions. We acknowledge financial support from the Swiss NCCR
``Quantum Photonics" and National Science Foundation (SNSF), the
EU Qubit Applications Project (QAP) Contract number 015848, the
Spanish MEC Consolider QOIT and FIS2007-60182 projects and a
``Juan de la Cierva" grant, la Generalitat de Catalunya and Caixa
Manresa, the National Research Foundation and the Ministry of
Education, Singapore, and the IAP project Photonics@be of the
Belgian Science Policy.

\addcontentsline{toc}{section}{References}
\bibliographystyle{unsrt}
\bibliography{DIQKD}

\end{document}